\documentclass[letterpaper, 11pt]{article}

\usepackage{hyperref}
\usepackage{amsthm}
\usepackage{amssymb}
\usepackage{amsmath}
\usepackage{xcolor}
\usepackage{bbm}
\usepackage[utf8]{inputenc}
\usepackage[top=1in, bottom=1in, left=1in, right=1in]{geometry}

\newtheorem{thm}{Theorem}[section]
\newtheorem{lem}[thm]{Lemma}
\newtheorem{definition}[thm]{Definition}
\newtheorem{cor}[thm]{Corollary}
\newtheorem{prop}[thm]{Proposition}
\newtheorem{rem}[thm]{Remark}

\numberwithin{equation}{section}
\title{Optimal Fast Johnson-Lindenstrauss Embeddings for Large Data Sets}

\newcommand{\footremember}[2]{
	\footnote{#2}
	\newcounter{#1}
	\setcounter{#1}{\value{footnote}}
}
 
\author{
	Stefan Bamberger\footremember{tum}{stefan.bamberger@tum.de, Department of Mathematics, Technische Universität München}
	\and Felix Krahmer\footremember{tum2}{felix.krahmer@tum.de, Department of Mathematics, Technische Universität München}
}

\date{\today}

\DeclareMathOperator{\Var}{Var}
\DeclareMathOperator{\Cov}{Cov}

\begin{document}
\pagenumbering{gobble}
\thispagestyle{empty}
	
\maketitle
\begin{abstract}
	Johnson-Lindenstrauss embeddings are widely used to reduce the dimension and thus the processing time of data. To reduce the total complexity, also fast algorithms for applying these embeddings are necessary. To date, such fast algorithms are only available either for a non-optimal embedding dimension or up to a certain threshold on the number of data points.

	We address a variant of this problem where one aims to simultaneously embed larger subsets of the data set. Our method follows an approach by Nelson: A subsampled Hadamard transform maps points into a space of lower, but not optimal dimension. Subsequently, a random matrix with independent entries projects to an optimal embedding dimension.
	
	For subsets whose size scales at least polynomially in the ambient dimension, the complexity of this method comes close to the number of operations just to read the data under mild assumptions on the size of the data set that are considerably less restrictive than in previous works. We also prove a lower bound showing that subsampled Hadamard matrices alone cannot reach an optimal embedding dimension. Hence, the second embedding cannot be omitted.
\end{abstract}
	
	

\section{Introduction}
\subsection{Johnson-Lindenstrauss Embeddings and Applications}

Dimension reduction has played an increasingly significant role in data science in recent years due to the increasing size and dimensionality of available data. So-called Johnson-Lindenstrauss embeddings (JLEs) are of central importance in this context. Such maps reduce the dimension of an arbitrary finite set of points while preserving the pairwise distances up to a small deviation. JLEs were first introduced in \cite{jloriginal} in the context of Banach spaces. They continue to be intensively studied in both mathematics and computer science.

For reasons of practicability, most works focus on linear maps based on random constructions, aiming to be independent of the data to be considered. In this spirit, we work with the following definition of a Johnson-Lindenstrauss embedding ($E$ can be thought of as the set of pairwise distances).

\begin{definition}[Johnson-Lindenstrauss Embedding] \label{def_jle}
	Let $A \in \mathbb{R}^{m \times N}$ be a random matrix where $m < N$, $\epsilon, \eta \in (0, 1)$ and $p \in \mathbb{N}$. We say that $A$ is a $(p, \epsilon, \eta)$-JLE (Johnson-Lindenstrauss embedding) if for any subset $E \subseteq \mathbb{R}^{N}$ with $|E| = p$
	\begin{equation}
	(1 - \epsilon) \|x\|_2^{2} \leq \|A x\|_2^{2} \leq (1 + \epsilon) \|x\|_2^{2}  \label{jlp_ineq}
	\end{equation}
	holds simultaneously for all $x \in E$  with a probability of at least $1 - \eta$.
\end{definition}

The original work of Johnson and Lindenstrauss constructed $(p, \epsilon, \eta)$-JLEs with an embedding dimension of $m \geq C_\eta \epsilon^{-2} \log(p)$, which has recently been shown to be order optimal under minimal assumptions \cite{jl_optimal}. Various subsequent works developed simplified approaches for constructing JLEs. Notably, Achlioptas \cite{binaryjl} considered a  matrix $A \in \mathbb{R}^{m \times N}$ with independent entries $A_{j k} \in \{\pm 1\}$ satisfying $\mathbb{P}(A_{j k} = 1) = \mathbb{P}(A_{j k} = -1) = \frac{1}{2}$. In this case the normalized matrix $\frac{1}{\sqrt{m}} A$ is a $(p, \epsilon, \eta)$-JLE in the above setting, again for order optimal embedding dimension. Dasgupta and Gupta \cite{gupta_dasgupta_gauss} showed the same for a matrix with independent normally distributed (Gaussian) entries.

In many applications of Johnson-Lindenstrauss embeddings, even such simplified constructions are of limited use due to the trade-off between the complexity benefit for the original problem resulting from the dimension reduction and the additional computations required to apply the JLE.

Ailon and Chazelle \cite{fastjl} addressed this issue in connection to the approximate nearest neighbor search problem (ANN), the problem of finding a point $\hat{x}$ in a given finite set $E \subseteq \mathbb{R}^N$ such that one has $\|x - \hat{x}\| \leq (1 + \epsilon) \min_{v \in E} {\|x - v\|}$ for a given $x \in \mathbb{R}^N$.

Their algorithm uses a preprocessing step that transforms all $p$ points in $E$ using a JLE. This step is known to have a high computational complexity, but as it can be performed offline, it is not considered to be the main computational bottleneck.
Rather the focus of the analysis has been on  the subsequent query steps, in which the JLE is applied to new inputs $x$, reducing the dimension for the subsequent computations. For this purpose, the authors design the so-called Fast Johnson-Lindenstrauss Transform (FJLT) whose application to $x$ requires a particularly low computation time. 
It should be noted, however, that if only one or very few query steps are to be performed, the problem is likely to be feasible even without the FJLT -- a challenging scenario of interest will be to embed a larger number of points. In the query step, these are typically significantly fewer than in the full data set; at the same time, for the preprocessing step, also a fast embedding of the entire set may be of interest.

In other applications, what would correspond to the preprocessing step in \cite{fastjl} is the central task of interest. That is, the data set to be embedded is given in full and the goal is to efficiently compute the dimension reduction of the whole set at once.

For example, such a setup appears in various approaches for nonlinear dimension reduction such as Isomap \cite{isomap}, Locally Linear Embedding \cite{lle} or Laplacian eigenmaps \cite{lap_eig}. These methods exploit that a high-dimensional data set lying near a low-dimensional manifold can be locally approximated by linear spaces. Hence it is of key importance to identify the data points close to a point of interest in the data set, as the linear approximation will be valid only for those. Consequently, one only works with points in the data set, no new points enter at a later stage, and it suffices to apply the JLE simultaneously to the whole data set before searching for multiple approximate nearest neighbors within its projection.

A simultaneous fast transformation of the entire data set is also a central goal for various applications in randomized numerical linear algebra. The following lemma, for example, introduces a randomized approach for approximate matrix multiplication by multiplying matrices $\hat{A} = (S A^*)^*$ and $\hat{B} = S B$ shrunk by a Johnson-Lindenstrauss embedding $S$ rather than computing the product of the full matrices.
\begin{lem}[Lemma 6 in \cite{mat_alg}] \label{lem_randmat}
	Let $A \in \mathbb{R}^{q \times N}, B \in \mathbb{R}^{N \times p}$ be matrices and $S \in \mathbb{R}^{m \times N}$ a $(q + p, \epsilon, \eta)$-JLE. Then
	\[
	\mathbb{P}(\|AB - AS^*SB\|_F \leq \epsilon \|A\|_F \|B\|_F) \geq 1 - \eta.
	\]
\end{lem}

Again, the JLE is applied to the whole matrices $A^*$ and $B$ and no additional data points enter at a later stage.

The computation time required for this approximate product consists, on the one hand, of the time $\mathcal{O}(q m p)$ for the reduced product and, on the other hand, of the computational effort to calculate $S B$ and $S A^*$, i.e.~to apply the JLE to the entire data set (the columns of $B$ and the rows of $A$). 

If the JLE used does not admit a fast transform and standard matrix multiplication is applied,
its computational cost can easily dominate the computation time for the whole approximate product.
Namely, assuming w.l.o.g.~that $q \leq p$, applying $S$ to the data set requires $\mathcal{O}(m N q + m N p) = \mathcal{O}(m N p)$ operations. If $q \leq N$, this becomes the dominant part of the computation, if $q \leq m$, it even surpasses the complexity of the original multiplication $\mathcal{O}(q N p)$.

With a fast Johnson-Lindenstrauss transform, this changes and the cost of the multiplication in the embedding space will typically become dominant. Namely, if applying the transform requires say $\mathcal{O}(N \log N)$ operations per data point, the total cost for the JLE step is $\mathcal{O}(p N \log N)$, which will be less than $\mathcal{O}(m N p)$ in basically all interesting cases.  Thus minimizing the total computational cost boils down to choosing the embedding dimension $m$ as small as possible.

Given that the Fast Johnson-Lindenstrauss Transform by Ailon and Chazelle \cite{fastjl} and its extension by Ailon and Liberty \cite{fastjl2} yield order-optimal embedding dimensions, they provide near-optimal solutions to these problems whenever the regime is admissible. 
However, these works only apply to data sets of size $p \leq \exp(\mathcal{O}_\epsilon(N^{\frac{1}{3}}))$ or $p \leq \exp(\mathcal{O}_\epsilon(N^{\frac{1}{2}-\delta}))$ for some arbitrary $\delta>0$ (that will impact the constants appearing in the embedding dimension), respectively. Here the notation $\mathcal{O}_\epsilon$ indicates that the involved constants may depend on $\epsilon$. Follow-up works without such restrictions (e.g., \cite{alm_opt_fjlt}, \cite{riptojl}) require additional logarithmic factors in the embedding dimension, which can directly impact the total computational cost as, e.g., in the second example discussed. 

To our knowledge, there were no approaches known before this work of applying Johnson-Lindenstrauss embeddings with an optimal embedding dimension to data sets of size larger than $\exp(\mathcal{O}_\epsilon(N^{\frac{1}{2}}))$ which admit a fast transform, neither for an individual query step nor for transforming larger subsets or the full data set simultaneously.

\subsection{Contributions of this Work}

The goal of this note is to provide such methods for the case of simultaneous embeddings of larger subsets of the data sets or the entire set.
To do this, we consider a class of matrices based on a composition of a dense matrix with random signs and a randomly subsampled Hadamard transform. Such embeddings have been considered in the literature before, for example in \cite{nelson_rip_fast}.

Our contribution is the following. Assuming only that  $p = \exp(\mathcal{O}_\epsilon(N^{1 - \delta}))$ for some $\delta > 0$, we provide a method for simultaneously applying this map to $p'$ points in time $\mathcal{O}(p' N \log m)$ where $p' \geq N^{d(\delta)}$ and the exponent $d(\delta)$ depends only on $\delta$. Up to the logarithmic factor, this complexity is just the number of operations required to read the data. The assumption on $p$ is indeed mild; it admits data sets of size close to $p=\exp(\Theta_\epsilon(N))$, for which the identity provides a trivial order optimal embedding, so our result almost unrestricted for large $p$.

Note that this statement includes a fast transformation of the entire data set $E$, but it is more general. In particular, with explicit bounds for $d(\delta)$ for certain choices of $\delta$, we obtain fast simultaneous embeddings of as few as $N$ points and data sets of admissible size significantly beyond state-of-the-art results.

With slight modifications this can also be made to work for embeddings into a space equipped with the $\ell_1$ instead of the $\ell_2$ norm. We also provide a lower bound on the possible embedding dimension of subsampled Hadamard transforms without the composition with the dense matrix. This shows that they cannot provide an optimal embedding dimension as desired in this paper.

\subsection{Outline}

We start by reviewing the construction of the Fast Johnson-Lindenstrauss Transform in Section \ref{subsec_fjlt}. Section \ref{subsec_jl_rip} discusses previous work on alternative fast Johnson-Lindenstrauss embeddings with no restriction on $p$, but an embedding dimension that is suboptimal by logarithmic factors.
Section \ref{sec_comp} then presents the composed embedding used for the analysis and our main contribution, a method for fast (simultaneous) transformations of large data sets without an essential restriction on $p$.
The proof of our main result regarding the complexity is presented in Section \ref{sec:complexity}. In the subsequent Section \ref{sec:ell1} we show that this method can be adapted for fast embeddings into spaces with the $\ell_1$ norm. In Section \ref{sec:lower_bounds} we show the lower bound on the embedding dimension for subsampled Hadamard matrices. Then we conclude with a discussion in Section~\ref{sec:discussion}.

\section{Background and Previous Work} \label{sec_existingjl}

\subsection{Notation and Basic Terms} \label{sec:notation_basic}

We make use of the standard $\mathcal{O}$ notation, i.e.,~$f(x) = \mathcal{O}(g(x))$ if $f(x) \leq C g(x)$ for a constant $C > 0$ and all $x > x_0$ or $0 < x < x_0$ for fixed $x_0$ depending on whether the limit $x \rightarrow \infty$ or $x \rightarrow 0$ is considered, $f(x) = \Omega(g(x))$ analogously for $f(x) \geq C g(x)$ and $f(x) = \Theta(g(x))$ if both $f(x) = \mathcal{O}(g(x))$ and $f(x) = \Omega(g(x))$ at the same time. We write $\mathcal{O}_c$ for a parameter $c$ if the corresponding $x_0$ and $C$ can depend on $c$, analogously for $\Omega_c$ and $\Theta_c$. Sometimes we also write $f(x) \lesssim g(x)$ instead of $f(x) = \mathcal{O}(g(x))$ and $f(x) \asymp g(x)$ for $f(x) = \Theta(g(x))$.

We write $f(x) = o(g(x))$ for $g(x) > 0$ if $\lim \frac{f(x)}{g(x)} = 0$, with the limit depending on the context.

We use the notion of subgaussian random variables $X$ and two equivalent characterizations:
\begin{itemize}
	\item $\mathbb{P}(|X| > t) \leq \exp(1 - t^2 / K_1^2)$
	\item $\sup_{p \geq 1} \mathbb{E}[|X|^p] \leq K_2$
\end{itemize}
For the proof of equivalence and more characterizations, see \cite{vershynin_rand_mat} for a summary. Every subgaussian random variable satisfies both conditions with $K_1$ and $K_2$ that differ from each other by at most a constant factor. We define the corresponding $K_1$ to be the subgaussian norm $\| X \|_{\psi_2}$.

We call a vector $\xi \in \{\pm 1\}^N$ a Rademacher vector if it has independent entries and each of them takes the value $\pm 1$ with a probability of $\frac{1}{2}$ each.

We use $H \in \mathbb{R}^{N \times N}$ with $N = 2^n$ a power of $2$ for the Hadamard transformation on $\mathbb{R}^N$. This is defined as the $n$-times Kronecker Product of $H_1 := \frac{1}{\sqrt{2}} \begin{pmatrix} 1 & 1 \\ 1 & -1 \end{pmatrix}$ with itself, i.e.~$H = \bigotimes_{j = 1}^n H_1$. $H_1$ is the Fourier transform on the group $Z_2$, thus we can also regard $H$ as the Fourier transform on the group $Z_2^n$ by fixing a bjiection between $Z_2^n$ (the same as $\mathbb{F}_2^n$) and $[N]$. Then for $u, v \in \mathbb{F}_2^n$, the corresponding entry of the matrix is also given by $H_{u v} = (-1)^{\langle u, v \rangle}$ where $\langle \cdot \rangle$ denotes the inner product in $\mathbb{F}_2^n$.

The matrix $H$ is unitary and all its entries have an absolute value of $\frac{1}{\sqrt{N}}$. Additionally, there is a fast algorithm (Walsh-Hadamard transform) which can compute the matrix vector product $H x$ for any $x \in \mathbb{R}^N$ in time $\mathcal{O}(N \log N)$. For details about the algorithm, see Section 7 of \cite{fastjl2}.

There are also other matrices in complex numbers which have these three properties such as the discrete Fourier transform. The main results of this paper can also be adapted to these matrices. However, we are going to use the Hadamard transform since it ensures that our embedding stays real valued.

For counterexamples we are also going to use specific structural properties of the Hadamard matrix. Denote $\mathbb{G}_{n, r}$ for the set of all $r$-dimensional subspaces of the vector space $\mathbb{F}_2^n$. For a subset $V \subseteq \mathbb{F}_2^n$ we denote $\mathbbm{1}_V \in \mathbb{R}^{N}$ for the indicator vector corresponding to $V$ with the normalization $\| \mathbbm{1}_V \|_2 = 1$. We use the following fact about subspaces: If $V \in \mathbb{G}_{n, r}$, then 
\begin{equation}
H \mathbbm{1}_V = \mathbbm{1}_{V^\perp} \label{hadamard_subspace}
\end{equation}
(see for example \cite{lower_bound_ieee}, Lemma II.1).

\subsection{Optimal Fast Johnson-Lindenstrauss Embeddings} \label{subsec_fjlt}

The Fast Johnson-Lindenstrauss Transformation (FJLT) by Ailon and Chazelle \cite{fastjl} consists of a fast Hadamard transform combined with a sparse projection. More precisely, the embedding $A = P H D_\xi \in \mathbb{R}^{m \times N}$ is defined as follows:
\begin{itemize}
	\item $P \in \mathbb{R}^{m \times N}$ is a sparse projection with independent entries $P_{j k} = b_{j k} g_{j k}$ where all $b_{j k}$ and $g_{j k}$ are independent and $\mathbb{P}(b_{j k} = 1) = q$, $\mathbb{P}(b_{j k} = 0) = 1 - q$ and $g_{j k} \sim N(0, (m q)^{-1})$ for a $q = \min\left\{ \Theta\left( \frac{(\log p)^2}{N} \right), 1 \right\}$.
	
	\item $H \in \mathbb{R}^{N \times N}$ is a full Hadamard matrix.
	
	\item $D_\xi \in \mathbb{R}^{N \times N}$ is a diagonal matrix corresponding to a Rademacher vector $\xi \in \{1, -1\}^N$.
\end{itemize}

Lemma 1 in \cite{fastjl} states that this random matrix is a $(p, \epsilon, \frac{2}{3})$-JLE for $m = \Theta(\epsilon^{-2} \log p)$. Since the transformation by $H$ can be computed in time $\mathcal{O}(N \log N)$ (using the Walsh-Hadamard transform described e.g.~in \cite{fastjl2}) and $P$ is sparse with a high probability, it is also shown that the entire transformation needs
\begin{equation}
\mathcal{O}\left(N \log N + \min\{N \epsilon^{-2} \log p, \epsilon^{-2} (\log p)^3 \} \right) \label{time_fastjl}
\end{equation}
operations with a probability of at least $\frac{2}{3}$. This is of order $\mathcal{O}(N \log N)$ provided $\log p = \mathcal{O}_\epsilon(N^\frac{1}{3})$. Furthermore, the embedding dimension exhibits optimal scaling in $N$, $p$ and $\epsilon$ \cite{jl_optimal}.

A further improvement of this approach is achieved by Ailon and Liberty in \cite{fastjl2} where the bound for the fast transformation is raised from $\log p = \mathcal{O}_\epsilon(N^{\frac{1}{3}})$ to $\log p = \mathcal{O}_\epsilon(N^{\frac{1}{2} - \delta})$ for any fixed $\delta > 0$. The computation time is also lowered to $\mathcal{O}(N \log m)$. 

\subsection{Unrestricted Fast Johnson-Lindenstrauss Embeddings} \label{subsec_jl_rip}

A different approach for the construction of Johnson-Lindenstrauss matrices was introduced by Ailon and Liberty in \cite{alm_opt_fjlt}. Compared to the aforementioned result, this construction does not have a significant restriction on the number $p$ of points in $E$ for a fast transformation. However, its embedding dimension has an additional polylogarithmic factor in $N$ as well as a suboptimal dependence on $\epsilon$. The dependence on $\epsilon$ is improved in \cite{riptojl}, making the construction optimal up to logarithmic factors in $N$.

Both constructions are based on the restricted isometry property (RIP) of the embedding matrix, that is, approximate norm preservation of vectors with many vanishing entries.

\begin{definition}[Restricted Isometry Property (RIP)]
	A matrix $\Phi \in \mathbb{R}^{m \times N}$ has the $(k, \delta)$-RIP (restricted isometry property) if one has
	\[
	(1 - \delta) \|x\|_2^2 \leq \|\Phi x\|_2^2 \leq (1 + \delta) \|x\|_2^2
	\]
	for all $k$-sparse $x \in \mathbb{R}^N$, i.e.~for all $x \in \mathbb{R}^N$ that have at most $k$ non-zero entries.
\end{definition}

The following theorem of \cite{riptojl} shows that RIP matrices can be made to Johnson-Lindenstrauss transforms by randomizing their column signs.
\begin{thm}[Theorem 3.1 from \cite{riptojl}] \label{jl_rip_thm}
	Let $\epsilon, \eta \in (0, 1)$ and $E \subseteq \mathbb{R}^N$ with $|E| = p < \infty$. Assume that $\Phi \in \mathbb{R}^{m \times N}$ has the $(k, \delta)$-RIP for $k \geq 40 \log \frac{4p}{\eta}$ and $\delta \leq \frac{\epsilon}{4}$. Let $\xi \in \{1, -1\}^N$ be a Rademacher vector.
	
	Then $\Phi D_\xi$ is a $(p, \epsilon, \eta)$-Johnson-Lindenstrauss embedding.
\end{thm}

Using an argument based on Gelfand widths, one can show that for sufficiently small $\delta$, for a $(k, \delta)$-RIP matrix in $\mathbb{R}^{m \times N}$, one must have $m = \Omega(k \log(\frac{N}{k}))$ (see Chapter 10 in \cite{csbasic}). Thus Theorem \ref{jl_rip_thm} can only yield the Johnson-Lindenstrauss property for an embedding dimension of at least $m = \Omega(\epsilon^{-2} \log(p) \log(\frac{N}{k}))$ while the optimal one is independent of $N$.

Theorem \ref{jl_rip_thm} can be applied to the transform $R H D_\xi$ with a Rademacher vector $\xi \in \mathbb{R}^N$, a Hadamard matrix $H \in \mathbb{R}^{N \times N}$ and a random projection (subsampling) $R \in \mathbb{R}^{m \times N}$ selecting $m$ entries uniformly with replacement and rescaling by $\sqrt\frac{N}{m}$.

As $R H \in \mathbb{R}^{m \times N}$ is shown to have the $(k, \delta)$-RIP with high probability for $m = \Theta(\delta^{-2} k (\log N)^4)$ in \cite{fourierrip_n4}, this construction requires an embedding dimension of $m = \Theta(\epsilon^{-2} \log(\frac{p}{\eta}) (\log N)^4)$ which is optimal up to a factor of $(\log N)^4$. At the same time, the construction admits a fast computation of the matrix vector product even for large values of $p$: Using the fast Walsh-Hadamard transform, one can always achieve a computation time of $\mathcal{O}(N \log m)$ per data point.

Similar results can also be obtained for subsampled random convolutions instead of subsampled Hadamard transform (\cite{prc_rrt}, \cite{prc_suprema_chaos}).

Improved embedding dimensions can be obtained using new RIP bounds for partial Hadamard matrices (more generally, subsampled unitary matrices with entries bounded by $\mathcal{O}(\frac{1}{\sqrt{N}})$) that have been shown in \cite{fourierrip}.
Namely, it is shown that $m = \Omega(\delta^{-2} (\log\frac{1}{\delta})^2 k (\log\frac{k}{\delta})^2 \log(N))$ implies the $(k, \delta)$-RIP with probability arbitrarily close to $1$ for sufficiently large $N$. Thus, for $k \leq N$ and a fixed $\delta$, $m = \Omega(k (\log N)^3)$ is sufficient. So for a fixed $\epsilon$, this implies that the JLE construction introduced above needs an embedding dimension of $m = \Theta(\log(p) (\log N)^3)$, improving the previous result by a factor of $\log N$. However, for reasons of simplicity of presentation due to the simpler dependence on $\epsilon$, we continue to work with the bound $m = \Theta(\epsilon^{-2} \log(\frac{p}{\eta}) (\log N)^4)$ and leave the possible refinement to the reader.

\subsection{Rectangular Matrix Multiplication} \label{sec:rectmm}
We denote by $M_E \in \mathbb{R}^{N \times p}$ the matrix with the vectors in $E$ as columns. Then simultaneously applying a transform $A \in \mathbb{R}^{m \times N}$ to all vectors in $E$  corresponds to computing the matrix product $A M_E$.

Thus the question of computing the simultaneous embedding closely connects to the topic of fast algorithms for matrix multiplication, which, starting with the seminal work of Strassen \cite{strassenalg}, developed into a core area of complexity theory (see, for example, Chapter 15 of \cite{alg_comp_th} for an overview). Given that we are interested in dimension reduction, algorithms for the fast multiplication of rectangular matrices will be of particular relevance to this paper. When one of the two matrices is square and the other has many more columns, an asymptotically optimal family of algorithms has been provided by Lotti and Romani \cite{limitfastmm}.
More precisely, their result can be summarized in the following proposition.

\begin{prop} \label{prop_fastmm}
	There exists an increasing function $f: [1,\infty) \rightarrow [2,\infty)$ such that for each real-valued exponent $s \geq 1$, the multiplication of an $n \times n$ and an $n \times \lceil n^s \rceil$ matrix can be computed in time $\mathcal{O}(n^{f(s)})$ and $f(s)-s$ is decreasing and converges to $1$ for $s\rightarrow \infty$.  
\end{prop}

That is, in the limit, the number of operations will just be what is required to read the two matrices.

\section{Main Result} \label{sec_comp}
Unless otherwise noted, in all theorems we assume that $\epsilon \in (0, 1)$ and $\eta \in (0, \frac{1}{2})$ are arbitrary and that $m < N$.

A helpful and already known observation for creating Johnson-Lindenstrauss embeddings with good embedding dimension and fast multiplication is that the composition of two JLEs is again a JLE. We include a proof for completeness.

\begin{lem} \label{thm_comp}
	Let $A \in \mathbb{R}^{m \times n}$, $B \in \mathbb{R}^{n \times N}$ be independent random matrices that
	are both $(p, \frac{\epsilon}{3}, \frac{\eta}{4})$-JLEs.
	Then $A B \in \mathbb{R}^{m \times N}$ is a $(p, \epsilon, \eta)$-JLE.
\end{lem}

\begin{proof}
	Let $E \subseteq S^{N-1}$ with $|E| = p$. The probability that the norms in $B E$ are distorted by more than $1 \pm \frac{\epsilon}{3}$ compared to $E$ is $\leq \frac{\eta}{4}$.
	
	For each fixed value of $B$ that satisfies this norm preservation, the probability that a norm in $A (B E)$ is distorted outside of the range $[(1 - \frac{\epsilon}{3})^2, (1 + \frac{\epsilon}{3})^2]$ is less than $\frac{\eta}{4}$.
	
	By a union bound, with a probability of at least $1 - \eta$ all norms in $A B E$ lie within $[(1 - \frac{\epsilon}{3})^2, (1 + \frac{\epsilon}{3})^2] \subseteq [1 - \epsilon, 1 + \epsilon]$ as $\epsilon < 1$.
\end{proof}

This lemma can be used to combine the advantages of fast JLEs with non-optimal embedding dimension (as in Theorem \ref{jl_rip_thm}) and dense random matrices with optimal embedding dimension similar to those of Achlioptas \cite{binaryjl}. A set of $p$ vectors in $\mathbb{R}^N$ is first mapped into a space of non-optimal, but reduced dimension using a fast Johnson-Lindenstrauss transform and then a dense matrix maps the vectors from this space of reduced dimension to a space of optimal embedding dimension. Note that this leads to a smaller dense matrix and thus a faster computation compared to a JLE consisting of only a dense matrix.

We apply this procedure to the fast transform from \cite{riptojl} and the dense $\pm 1$ matrix from \cite{binaryjl}, obtaining a JLE of optimal embedding dimension. The usage of this construction has already been suggested for example in \cite{nelson_rip_fast}, see footnote 2.

\begin{cor} \label{totfjl}
	Consider the following matrices.
	\begin{itemize}
		\item $G \in \{\pm \frac{1}{\sqrt{m}} \}^{m \times n}$ where the entries $G_{j k}$ are independent with 
		\[
		\mathbb{P}(G_{j k} = \frac{1}{\sqrt{m}}) = \mathbb{P}(G_{j k} = -\frac{1}{\sqrt{m}}) = \frac{1}{2}.
		\]
		
		\item $R \in \mathbb{R}^{n \times N}$ is a random projection that selects $n$ of the $N$ entries of a vector uniformly at random with replacement and rescales by $\sqrt\frac{N}{n}$.
		
		\item $H \in \mathbb{R}^{N \times N}$ is a full Hadamard matrix.
		
		\item $D_\xi \in \mathbb{R}^{N \times N}$ where $\xi \in \{1, -1\}^N$ is a Rademacher vector.
	\end{itemize}
	
	Let $G$, $R$ and $\xi$ be independent. If $m \geq c_1 \epsilon^{-2} \log\frac{p}{\eta}$ and $n \geq c_2 \epsilon^{-2} \log(\frac{p}{\eta}) (\log N)^4$ for suitable absolute constants $c_1$ and $c_2$, then $G R H D_\xi \in \mathbb{R}^{m \times N}$ is a $(p, \epsilon, \eta)$-JLE.
\end{cor}

Note that the resulting construction is very similar to the one introduced in \cite{fastjl}. Namely, similarly to the matrix $P$ in the construction of \cite{fastjl} (cf.~Section \ref{subsec_fjlt} above), the matrix $\tilde{P} = G R$ has many zero entries and independent random entries in all other locations. The main difference is that in this construction the non-zero entries are concentrated in just a few columns while in \cite{fastjl}, their locations are random. As captured by the following main theorem, this fact leads to a significant speed up when simultaneously projecting large enough point sets, as the structure allows for the use of fast matrix multiplication.

\begin{thm} \label{thm_complexity}
	Consider $E \subseteq \mathbb{R}^N$ of size $p = \exp(\mathcal{O}_\epsilon(N^{1 - \delta}))$ for an arbitrary $\delta > 0$ and a failure probability $\eta \in (p^{-c}, \frac{1}{2})$ for constant $c$. Then for embedding dimensions $m = \Theta(\epsilon^{-2} \log\frac{p}{\eta})$ and $n = \Theta( m (\log N)^4)$, the transformation $A = G R H D_\xi$ introduced in Corollary \ref{totfjl} is a $(p, \epsilon, \eta)$-JLE for the set $E$.
	
	There is an exponent $d(\delta)$ depending only on $\delta$ such that the following holds. For any finite set $E' \subseteq \mathbb{R}^{N}$ with $p' = |E'| \geq N^{d(\delta)}$ or $p' = p$, $A$ can be applied simultaneously to all points in $E'$ in time $\mathcal{O}_\delta(p' N \log m)$.
\end{thm}

The dependence of the exponent $d(\delta)$ on $\delta$ is not explicitly stated in this theorem. This is caused by Proposition~\ref{prop_fastmm} (and the source \cite{limitfastmm}) not stating a bound on the convergence speed. However, certain estimates for specific values of the function $f$ in Proposition~\ref{prop_fastmm} have been made before and we use those to get explicit estimates for $d(\delta)$ in Section~\ref{sec:exponent}. This will lead to the following subsequent theorem summarizing two particularly interesting cases.

\begin{thm} \label{thm:example_numbers}
	Consider the setup of Theorem~\ref{thm_complexity}.
	\begin{itemize}
		\item If $p = \exp(\mathcal{O}_\epsilon(N^\frac{3}{4}))$ and $p' \geq N$, then $A$ can be applied to all points in $E'$ in time $\mathcal{O}(p' N \log m)$.
		
		\item If $p' \geq N^4$, then the transformation of $E'$ can be computed in $\mathcal{O}(p' N^{1.2})$ for any $p$ such that $m \leq N$ (that is, no further constraints are required on the size of $p$).
	\end{itemize}
\end{thm}

\section{Complexity Analysis -- Proof of Theorem \ref{thm_complexity}} \label{sec:complexity}

\subsection{Simultaneous Transformation of Data Points}

The main novelty in our approach is the use of Proposition \ref{prop_fastmm} in the last step of the embedding which is a (smaller) dense matrix $G$.
The proposition will be applied to provide a speedup for embedding dimensions $m \geq N^{\frac{1}{2} - \frac{\delta}{4}}$. The complementary case $m \leq N^{\frac{1}{2} - \frac{\delta}{4}}$ will be discussed in Section \ref{sec:singlep}, hence for the remainder of this section we will assume that $m \geq N^{\frac{1}{2} - \frac{\delta}{4}}$.

To prove Theorem \ref{thm_complexity}, we have to cover the cases $p' \geq N^{d(\delta)}$ and $p' = p < N^{d(\delta)}$. Consider the first case now. Recall that by assumption $p = \exp(\mathcal{O}_\epsilon(N^{1 - \delta}))$, more precisely we assume that $p \leq \exp(c' \epsilon^2 N^{1 - \delta})$ for a constant $c'$ chosen to ensure that $m \leq N^{1 - \delta}$.

To analyze the complexity, we note that computing the product consists of the following three steps. Recall that $M_{E'} \in \mathbb{R}^{N \times p'}$ is the matrix whose columns are all the vectors in the finite set $E'$.
\begin{itemize}
	\item Compute $M_1 := D_\xi M_{E'}$ by multiplying the diagonal matrix $D_\xi$ by each column of $M_{E'}$. This requires $\mathcal{O}(p' N)$ operations.
	
	\item Compute $M_2 := R H M_1$ by applying the trimmed Walsh-Hadamard transform to all columns of $M_1$ in $\mathcal{O}(p' N \log m)$ operations.
	
	\item Compute $M_3 := G M_2$ where $G \in \mathbb{R}^{m \times n}$ and $M_2 \in \mathbb{R}^{n \times p'}$. Choose $\alpha = \frac{\log p'}{\log m}$, yielding $p' = m^\alpha$. Split $G$ and $M_2$ up into blocks
	\[
	G = \begin{pmatrix}
	G_1 & \dots & G_r
	\end{pmatrix}
	\qquad
	M_2 = \begin{pmatrix}
	V_{1}\\
	\vdots\\
	V_{r}
	\end{pmatrix}
	\]
	for $r = \lceil \frac{n}{m} \rceil$ such that $G_1, \dots, G_r \in \mathbb{R}^{m \times m}$ and $V_{j} \in \mathbb{R}^{m \times p'}$ for $j \in [r]$. If necessary, pad the matrices with zero rows and columns such that the corresponding submatrices have the same size. Since $n \geq m$, this only changes the numbers of rows or columns by at most a constant factor.
	
	Each of the $r$ block multiplications required to compute this product is a multiplication of an $m \times m$ matrix by an $m \times m^\alpha$ matrix and thus requires $\mathcal{O}(m^{f(\alpha)})$ operations with $f$ as in Proposition \ref{prop_fastmm}. In the end, we need $\mathcal{O}(p' m r)$ operations to sum up the values for all entries of the result. In total, the number of operations is
	\[
	\mathcal{O}(r m^{f(\alpha)} + p' m r) = \mathcal{O}(p' n m^{f(\alpha) - \alpha - 1} + p' n				) = \mathcal{O}(p' n m^{f(\alpha) - \alpha - 1}),
	\]
	where in the last equality we used that by Proposition \ref{prop_fastmm} $f(\alpha) - \alpha \geq 1$.
\end{itemize}

Combining the runtime complexity of the three steps yields
\begin{equation}
\mathcal{O}(p' N \log(m) + p' n m^{f(\alpha) - \alpha - 1}) = \mathcal{O}\left(p' N \log(m) + p' m^{f(\alpha) - \alpha} (\log N)^4 \right).
\label{eqn_time}
\end{equation}

Proposition \ref{prop_fastmm} guarantees the existence of a $\beta = \beta(\delta)$ such that $\frac{1}{f(\beta) - \beta} > 1 - \delta$ and we obtain that
\[
\alpha = \frac{\log p'}{\log m} \geq \frac{d(\delta) \log N}{(1 - \delta) \log N} = \frac{d(\delta)}{1 - \delta}.
\]

Hence, for sufficiently large $d(\delta)$, one has that $\alpha \geq \beta$, i.e.,
\begin{equation}
(f(\alpha) - \alpha) (1 - \delta) < 1 \label{eq:falpha_alpha}
\end{equation}

Substituting this into (\ref{eqn_time}), using that $m = \mathcal{O}(N^{1 - \delta})$ yields a complexity of
\[
\mathcal{O}_\delta(p' N \log m + p' N) = \mathcal{O}_\delta(p' N \log m).
\]

Now consider the case that $m \geq N^{\frac{1}{2} - \frac{\delta}{4}}$ but $p' = p < N^{d(\delta)}$. Using these inequalities together with $m = \mathcal{O}(\epsilon^{-2} \log\frac{p}{\eta}) = \mathcal{O}(\epsilon^{-2} \log p)$ (the last equality follows from $\eta > p^{-c}$), we obtain that $\epsilon = \mathcal{O}\left(\frac{(d(\delta) \log N)^{\frac{1}{2}}}{N^{\frac{1}{4} - \frac{\delta}{8}}} \right)$. This case can be excluded by imposing the stronger condition that $p = \exp(\mathcal{O}(\epsilon^4 N^{1 - \delta}))$, which still is of the form $p = \exp(\mathcal{O}_\epsilon(N^{1 - \delta}))$. Namely, the two conditions for $\epsilon$ and $p$ would imply that 
\[
\log p = \mathcal{O}\left( \frac{(d(\delta) \log N)^2}{N^{1 - \frac{\delta}{2}}} N^{1 - \delta} \right) = \mathcal{O}( (d(\delta) \log N)^2N^{-\delta / 2} ),
\]
i.e., $p$ is bounded by a constant $\leq p_0(\delta)$. By definition of the $\mathcal{O}$ notation however, it is sufficient to show the complexity bound only for $p > p_0(\delta)$. This proves the theorem for the case of embedding dimensions $m \geq N^{\frac{1}{2}-\frac{\delta}{4}}$.

\begin{rem}
	Note that for small data sets of size $p \leq N^{d(\delta)}$ and \linebreak $\epsilon = \mathcal{O}\left(\frac{(d(\delta) \log N)^{\frac{1}{2}}}{N^{\frac{1}{4} - \frac{\delta}{8}}} \right)$ (the second case in the previous consideration), the increased exponent of $\epsilon$ in our condition for $p$ creates a gap between the applicability range of our result and the regime where the identity operator can be applied. Namely, our proof only applies to cases with embedding dimension at most $m = \mathcal{O}_\delta(\log(N) \cdot N^{\frac{1}{2} - \frac{\delta}{4}})$, while the identity cannot be used below $m=N$. To circumvent this, one can apply the construction of Ailon and Chazelle \cite{fastjl} to avoid this gap (as resulting from \eqref{time_fastjl}, the complexity is $\mathcal{O}(p N \log N + p \epsilon^{-2} (\log p)^3)$ which is $\mathcal{O}_\delta (p N (\log N)^2)$ in this case).
\end{rem}

\subsection{Transformation of Small Data Sets}\label{sec:singlep}

For $m \leq N^{\frac{1}{2} - \frac{\delta}{4}}$ (the case still missing in the previous section), we will just apply the embedding individually to each point in $E$. As noted in \cite{nelson_rip_fast}, the map defined in Corollary \ref{totfjl} is indeed a fast transform requiring
\begin{equation}
\mathcal{O}(N \log(m) + m^2 (\log N)^4) \label{comp_single}
\end{equation}
operations to be applied to a single data point using standard matrix multiplication.

Assuming again $p = \exp(\mathcal{O}(\epsilon^2 N^{\frac{1}{2} - \frac{\delta}{2}}))$, (\ref{comp_single}) is bounded by $\mathcal{O}_{\delta}(N \log(m))$ which yields a total complexity of $\mathcal{O}_\delta(p' N \log m)$, completing the proof of Theorem \ref{thm_complexity}.

\begin{rem}
	As noted in \cite{nelson_rip_fast}, if $p = \exp(\mathcal{O}(\epsilon^2 N^r))$ for $\frac{1}{2} \leq r \leq 1$, (\ref{comp_single}) becomes $\mathcal{O}(N^{2 r} (\log N)^4)$. This is still linear up to logarithmic factors when $r=\tfrac{1}{2}$ and close to linear for $r$ slightly larger than $\frac{1}{2}$. Thus, one obtains a fast transform also in these cases, in contrast to \cite{fastjl2} whose bound does not apply in either of these cases.
\end{rem}

\subsection{The Exponent $d(\delta)$ -- Proof of Theorem \ref{thm:example_numbers}} \label{sec:exponent}

While Theorem \ref{thm_complexity} applies for arbitrary subsets of size $p'$ bounded below by a polynomial in $N$, the exponent $d(\delta)$ remains unspecified since Proposition \ref{prop_fastmm} (i.e.,~\cite{limitfastmm}) does not quantify the convergence speed of $f(s) - s$. For certain values of $\delta$ explicit bounds for $d(\delta)$ can be derived from estimates for $f(s)$ available in the literature, such as those derived in \cite{fastmm_huang_pan} and subsequently improved in \cite{fastmm_others}.

More precisely, Theorem 5 in \cite{fastmm_others} adapted to the setup of Theorem \ref{prop_fastmm} states that
\begin{align}
f(s) \leq \displaystyle\min_{\substack{2 \leq k \in \mathbb{Z} \\ 0 \leq m_2, \beta, \beta' \in \mathbb{R}}} \frac{1}{\log(k^{m_1})} \log(\max\{N_1, N_2, N_3\}) \label{min_ub}
\end{align}
where
\begin{align*}
N_1 = g(m_1 + (\beta + \beta') m_2) g((1 + s) m_1) g(m_2), \\
N_2 = g(m_1 + (1 + \beta') m_2) g(2 m_1) g(\beta m_2), \\
N_3 = g(m_1 + (1 + \beta) m_2) g((1 + s) m_1) g(\beta' m_2), \\
m_1 = \frac{k + 2 - (1 + \beta + \beta') m_2}{2 + s} > 0, \\
g(x) = x^x\text{ for } x > 0 \text{ and } g(0) = 1.
\end{align*}

Some specific bounds for $f(s)$ resulting from specific choices for  $k, m_2, \beta, \beta'$ in (\ref{min_ub}) are given in Table \ref{table:expl_values}. To infer admissible values for $d(\delta)$, we recall that it is sufficient to satisfy the inequality (\ref{eq:falpha_alpha}) or equivalently

\[
f(\frac{d(\delta)}{1 - \delta}) - \frac{d(\delta)}{1 - \delta} - 1 < \frac{1}{1/\delta - 1}.
\]
and check it for the values given in Table \ref{table:expl_values}.

E.g., for $\delta = \frac{1}{4}$, by Table \ref{table:expl_values} it is sufficient to have $s = \frac{d(\delta)}{1 - \delta} = 1.333$ since $f(s) - s - 1  < \frac{1}{3}$. Solving for $d(\delta)$ we obtain that $d(\frac{1}{4}) = s (1 - \frac{1}{4}) = 0.99\dots < 1$ is admissible. Analogously we obtain that the following values are admissble. 
\begin{itemize}
	\item $d(\frac{1}{4}) = 0.99\dots$
	\item $d(\frac{1}{5}) = 3 (1 - \frac{1}{5}) = 2.4$
	\item $d(\frac{1}{6}) = 4 (1 - \frac{1}{6}) = 3.33\dots$
	\item $d(\frac{1}{10}) = 15(1 - \frac{1}{10}) = 13.5$
\end{itemize}

So the exponent $d(\delta)$ in the required lower bound on the number of points to be transformed at the same time stays reasonably small. In particular, requiring only $p' \geq N$ already guarantees a fast transformation up to $\log p = \mathcal{O}_\epsilon(N^\frac{3}{4})$, clearly surpassing the previous $N^\frac{1}{2}$ bound. This proves the first part of Theorem \ref{thm:example_numbers}.

To show the second statement of Theorem \ref{thm:example_numbers}, we observe that the bound (\ref{eqn_time}) for the complexity holds independently of the upper bound on $\log p$ which can lead to upper bounds on the complexity which are close to linear in $N$. From the assumption $p' \geq N^{4} \geq m^4$ we deduce that $\alpha \geq 4$ and hence by Table \ref{table:expl_values}, $f(\alpha) - \alpha \leq 1.1915$. Thus one obtains a runtime complexity of $\mathcal{O}(p' N \log N + p' m^{1.1915} (\log N)^4)$, which is bounded by $\mathcal{O}(p' N^{1.2})$; this holds for any size $p$ of the entire data set. This completes the proof of Theorem \ref{thm:example_numbers}.

\begin{table}
	
	\begin{center}	
		\begin{tabular}{c|c|c|c|c|c}
			$s$ & $k$ & $m_2$ & $\beta$ & $\beta'$ & $f(s) - s - 1 \leq \dots$ \\
			\hline
			$1.333$ & $6$ & $0.1794$ & $0.3307$ & $1$ & $0.3330$ \\
			$2$ & $7$ & $0.2545$ & $0.0303$ & $1$ & $0.2699$ \\
			$3$ & $7$ & $0.3046$ & $0.0284$ & $1$ & $0.2204$ \\ 
			$4$ & $8$ & $0.3326$ & $0.0147$ & $1$ & $0.1915$ \\
			$15$ & $21$ & $0.4114$ & $2.3734 \times 10^{-4}$ & $1$ & $0.1094$
		\end{tabular}
		
		\caption{Certain values substituted into the upper bound in (\ref{min_ub})}
		\label{table:expl_values}
	\end{center}
\end{table}

\section{Embeddings from $\ell_2$ to $\ell_1$} \label{sec:ell1}

Unlike the standard case for Johnson-Lindenstrauss embeddings with Euclidean norm considered in this work until here, it can also be interesting to consider embeddings from $\ell_2$ to $\ell_1$. To study these, we work with the following definition which is analogous to Definition \ref{def_jle}.
\begin{definition}[Johnson-Lindenstrauss Embedding for $\ell_2 \rightarrow \ell_1$]
	Let $A \in \mathbb{R}^{m \times N}$ be a random matrix, where $m < N$, $\epsilon, \eta \in (0, 1)$ and $p \in \mathbb{N}$. We say that $A$ is a $(p, \epsilon, \eta)$-JLE (Johnson-Lindenstrauss embedding) for $\ell_2 \rightarrow \ell_1$ if for any subset $E \subseteq \mathbb{R}^{N}$ with $|E| = p$
	\begin{equation}
	(1 - \epsilon) \|x\|_2 \leq \|A x\|_1 \leq (1 + \epsilon) \|x\|_2  \label{jlp_ineq_l1}
	\end{equation}
	holds simultaneously for all $x \in E$  with a probability of at least $1 - \eta$.
\end{definition}

Random matrices have already been used for embeddings into an $\ell_1$ space by Johnson and Schechtman \cite{johnson_schechtman_l1}. Ailon and Chazelle \cite{fastjl} also studied an embedding of this type to give an improved algorithm for the approximate nearest neighbor search problem. Their $\ell_1$ approach is a version of the FJLT described in Section~\ref{subsec_fjlt} that only uses a different scaling and a different probability parameter $q$. The resulting complexity for the transformation of a single point is
\[
\mathcal{O}(N \log N + \min\{\epsilon^{-2} N \log p, \epsilon^{-3} (\log p)^2\}).
\]

This implies that a fast transformation of $p'$ points in time $\mathcal{O}(p' N \log N)$ is only possible if $\log p = \mathcal{O}_\epsilon(N^{1/2})$. The same bound also applies to the subsequent improvement by Ailon and Liberty \cite{fastjl2}.

Without this restriction of $\log p = \mathcal{O}_\epsilon(N^{1/2})$, the best method available in the literature for embedding $p'$ points from $\ell_2^N$ to $\ell_1^m$ in time $\mathcal{O}(p' N^{1 + o(1)})$ is given by \cite{indyk_l1_embeddings} although this result applies to a significantly more general scenario of embedding the entire space $\ell_2^N$ into $\ell_1^m$. However, that result requires an embedding dimension $m = \Omega(N^{1 + o(1)})$ which is far from being optimal for the problem of embedding just finitely many points.

To extend our approach to embeddings into $\ell_1$, we recall that Gaussian matrices also yield $\ell_2 \rightarrow \ell_1$ JLEs. This well known fact arises for example as a special case in \cite{unified_l1}. We include a proof for completeness.

\begin{thm} \label{thm:l2l1jle}
	Let $G \in \mathbb{R}^{m \times N}$ be a random matrix with $G_{j k} \sim N(0, 1)$ and $m = \Omega(\epsilon^{-2} \log(\frac{p}{\eta}) )$ then $\sqrt\frac{\pi}{2} \frac{1}{m} G$ is a $(p, \epsilon, \eta)$-JLE for $\ell_2 \rightarrow \ell_1$.
\end{thm}

\begin{proof}
	Let $A := \sqrt\frac{\pi}{2} \frac{1}{m} G$. Fix $E \subseteq \mathbb{R}^N$ and $x \in E$. Then for $r = 1, \dots, m$, $(A x)_r \sim N(0, \frac{\pi \|x\|_2^2}{2 m^2})$.
	
	This means $\mathbb{E}[|(A x)_r|] = \frac{\|x\|_2}{m}$ and $X_r := |(A x)_r| - \mathbb{E}[|(A x)_r|]$ is subgaussian with norm $\left\| X_r \right\|_{\psi_2} \lesssim \left\| (A x)_r \right\|_{\psi_2} \lesssim \sqrt\frac{\pi}{2} \frac{\|x\|_2}{m}$.
	
	By rotation invariance of subgaussian variables (see Lemma 5.9 in \cite{vershynin_rand_mat}), we obtain $\left\| \sum_{r = 1}^{m} X_r \right\|_{\psi_2}^2 \lesssim \sum_{r = 1}^m \| X_r \|_{\psi_2}^2 \lesssim \sum_{r = 1}^m \frac{\pi \|x\|_2^2}{2 m^2} = \frac{\pi}{2 m} \|x\|_2^2$.
	
	Then
	\[
	\mathbb{P}\left(\left| \|A x\|_1 - \|x\|_2 \right| > \epsilon \|x\|_2 \right) = \mathbb{P}\left(\left| \sum_{r = 1}^{m} X_r \right| > \epsilon \|x\|_2 \right) \leq \exp(1 - C m \epsilon^{2})
	\]
	for a constant $C > 0$.
	
	Using a union bound over all $x \in E$, we get
	\[
	\mathbb{P}\left( \exists x \in E: \left| \|A x\|_1 - \|x\|_2 \right| > \epsilon \|x\|_2 \right) \leq p \exp(1 - C m \epsilon^{2}) = \exp(1 - C m \epsilon^{2} + \log(p)).
	\]
	This bound is smaller than $\eta$ for $m > \frac{1}{C} \epsilon^{-2} (1 + \log(\frac{p}{\eta}))$.
\end{proof}

Using this construction, the results of this paper can be generalized. Lemma \ref{thm_comp} in the same way also holds for the composition of an $\ell_2 \rightarrow \ell_2$ and an $\ell_2 \rightarrow \ell_1$ JLE with slightly different constants. Corollary \ref{totfjl} can be adjusted to the $\ell_1$ norm too. Note that this only holds if we replace the matrix $G$ in Corollary \ref{totfjl} by the matrix given in Theorem \ref{thm:l2l1jle} which has Gaussian entries instead of $\pm 1$ entries. Otherwise the equality $\mathbb{E}\|A x\|_1 = \|x\|_2$ for the expectation would not hold. Since Theorem \ref{thm_complexity} is not affected by this, it also holds in the $\ell_1$ case which enables the fast transformation of the entire data set or a subset of polynomial size. This is the main advantage of this method compared to the results given by Ailon, Chazelle \cite{fastjl} and by Ailon, Liberty \cite{fastjl2}.

\section{Lower Bounds for Subsampled Hadamard Matrices} \label{sec:lower_bounds}

As shown in the previous sections, the construction $G R H D_\xi$ from Corollary \ref{totfjl} provides JLEs into $\ell_2$ and $\ell_1$ with state-of-the-art embedding dimensions $m$ and up to a certain number of embedded points, this is also a fast pointwise embedding. For very large data sets, however, the application of the dense matrix $G$ prevents the complexity from staying below the desired threshold.

Thus, a natural question is whether this last embedding step is required or whether the matrix construction $R H D_\xi$ without the dense $G$ already yields JLEs for comparable joint embedding dimensions $m$.

In the following two sections, we provide negative answers to this question, both, for embeddings into $\ell_2$ and into $\ell_1$.

\subsection{A Lower Bound for Embeddings into $\ell_2$}

We will prove a lower bound on the possible embedding dimension necessary for $R H D_\xi$ to be a JLE into $\ell_2$, showing in particular that it has to depend on the ambient dimension $N$. Our approach is based on \cite{lower_bound_ieee} which provides a lower bound on the embedding dimensions for the restricted isometry property of Hadamard matrices. We will adapt this counterexample to the setting of Johnson-Lindenstrauss embeddings.

To this end, we use the properties of the Hadamard transform $H \in \mathbb{R}^{N \times N}$ and its interpretation as the Fourier transform on $\mathbb{F}_2^n$ mentioned in Section \ref{sec:notation_basic}.

Let now $s = 2^r$ be a power of two. Let $Q$ be the random multiset of row indices that the subsampling matrix $R$ selects. For an $r$-dimensional subspace $V$ of $\mathbb{F}_2^n$, we consider the indicator variable $X_V := \mathbbm{1}_{\{Q \cap V^\perp = \emptyset\}}$. The proof in \cite{lower_bound_ieee} considers 
\[
X := \sum_{V \in \mathbb{G}_{n, r}} X_V
\]
where $\mathbb{G}_{n, r}$ is the set of all $r$-dimensional subspaces of $\mathbb{F}_p^n$.

If $X \neq 0$, then there is a $V \in \mathbb{G}_{n, r}$ satisfying $Q \cap V^\perp = \emptyset$. Recall that by (\ref{hadamard_subspace}), it holds that $H \mathbbm{1}_V = \mathbbm{1}_{V^\perp}$. This implies that $R H \mathbbm{1}_V = R \mathbbm{1}_{V^\perp} = 0$ since $R$ only selects the entries indexed by $Q$.

In contrast to our way of selecting the subsampled entries $Q$, their result works with a random subset of indices $\tilde{Q} := \{j \in [N] | \delta_j = 1\}$ based on independent random selectors $\delta_j$ that are $1$ with probability $\frac{m}{N}$ and $0$ otherwise. Together with rescaling by factor $\sqrt\frac{N}{m}$, we denote this subsampling as $\tilde{R}$. Note that in contrast to the matrix $R$ used in the previous sections, the number of rows of $\tilde{R}$ might not exactly be $m$.

Upper bounding $\mathbb{P}(X = 0)$ from above leads to the following main result.

\begin{thm}[Theorem III.1 in \cite{lower_bound_ieee}] \label{lower_bound_main}
	Let $N$, $s$ and $\tilde{R}$ be as above.
	
	Assume $\min(r, n - r) \geq 12 \log_2 n$. Then there is a positive constant $C > 0$ such that if $m \leq C s \log(s) \log(\frac{N}{s})$, then with a probability of $1 - o(1)$ (for $N \rightarrow \infty$), there is a subspace $V \in \mathbb{G}_{n, r}$ which satisfies $\tilde{R} H \mathbbm{1}_V = 0$.
\end{thm}

However, with slight modifications, this result can also be used for the subsampling matrix $R$ instead of $\tilde{R}$. The key idea in \cite{lower_bound_ieee} used to prove Theorem \ref{lower_bound_main} is applying the second moment method (see e.g.~\cite{extremal_combinatorics}, Section 21.1) to $X$ yielding
\begin{align}
\mathbb{P}(X = 0) \leq \frac{\Var{X}}{(\mathbb{E} X)^2} = \frac{\sum_{d = d_0}^{n - r} \sum_{U, V: \dim(U^\perp \cap V^\perp) = d} \Cov(X_U, X_V)}{|\mathbb{G}_{n, r}|^2 (\mathbb{E} X_{V_0})^2 } \nonumber \\
= \frac{1}{|\mathbb{G}_{n, r}|^2} \sum_{d = d_0}^{n - r} \sum_{U, V: \dim(U^\perp \cap V^\perp) = d} \left[ \frac{\mathbb{E}(X_U X_V)}{ (\mathbb{E} X_{U})^2 } - 1 \right]  \label{moment2_method}
\end{align}
where $d_0 := \max(n - 2 r, 0)$ is the minimum dimension of $U^\perp \cap V^\perp$ and $V_0 \in \mathbb{G}_{n, r}$ is arbitrary. For the last equality, we used that $\mathbb{E} X_V$ is the same for all $V$.

The bounds on the first and second moments can be adapted to subsampling with replacement as used in $R$.

For any $V \in \mathbb{G}_{n, r}$, $V^\perp$ has $\frac{N}{s}$ elements and thus each subsampled index is not in $V^\perp$ with probability $1 - \frac{1}{s}$ and
\begin{align*}
\mathbb{E} X_V = \left(1 - \frac{1}{s} \right)^m.
\end{align*}

For any pair of subspaces $U, V \in \mathbb{G}^{n, r}$ with $\dim(U^\perp \cap V^\perp) = d$ we obtain
\begin{align*}
\mathbb{E}(X_U X_V) = \mathbb{P}(U^\perp \cap Q = \emptyset \wedge V^\perp \cap Q = \emptyset) = \left(1 - \frac{2}{s} + \frac{2^d}{N} \right)^m.
\end{align*}

Combining these expressions yields with $s \geq 2$,
\begin{align}
\frac{\mathbb{E}(X_U X_V)}{(\mathbb{E} X_U)^2} = \left( \frac{1 - \frac{2}{s} + \frac{1}{s^2} + \frac{2^d}{N} - \frac{1}{s^2} }{(1 - \frac{1}{s})^2} \right)^m = \left(1 + \frac{\frac{2^d}{N} - \frac{1}{s^2}}{(1 - \frac{1}{s})^2} \right)^m  \nonumber \\
= \left(1 + \left(\frac{s}{s - 1} \right)^2 \left( \frac{2^d}{N} - \frac{1}{s^2} \right) \right)^m \leq \left(1 + 4 \cdot \frac{2^d}{N} \right)^m \leq \exp\left( 2^d \cdot \frac{4 m}{N} \right). \label{variance_estimate}
\end{align}

Substituting (\ref{variance_estimate}) into (\ref{moment2_method}) leads to

\[
\mathbb{P}(X = 0) \leq \sum_{d = d_0}^{n - r} { \frac{T(n, r, d)}{|\mathbb{G}_{n, r}|^2} \left( \exp(2^d \cdot \frac{4 m}{N}) - 1 \right) }
\]
where $T(n, r, d)$ is the number of pairs $U, V \in \mathbb{G}_{n, r}$ such that $\dim(U^\perp \cap V^\perp) = d$.

This is the same bound found in \cite{lower_bound_ieee} with $\frac{m}{N}$ replaced by $\frac{4m}{N}$. 
Thus the proof of \cite{lower_bound_ieee} that $\mathbb{P}(X = 0) = o(1)$ carries over to the subsampling matrix $R$ as summarized in the following corollary.


\begin{cor} \label{lower_bound_adapted}
	Assume $s \geq 2$. The statement of Theorem \ref{lower_bound_main} also holds if we replace the subsampling with random selectors $\tilde{R}$ by subsampling with replacement $R$ as used in Corollary \ref{totfjl}.
\end{cor}

Now we can use this result to prove a lower bound for the embedding dimension of these matrices as Johnson-Lindenstrauss embeddings.

\begin{thm} \label{thm:had_lower_jl}
	For absolute constants $C, c > 0$ and arbitrary $\eta > 0$, there exists $N(\eta)$ such that the following holds. Let $N = 2^n \geq N(\eta)$ be a power of two and $R \in \mathbb{R}^{m \times N}$ the matrix representing independent random sampling with replacement and subsequent rescaling by $\sqrt\frac{N}{m}$. Let $p \geq p_0$ be a sufficiently large integer and assume it satisfies $(\log_2 N)^c \leq \log_2 p \leq \frac{N}{(\log_2 N)^c}$. Fix $\eta > 0$. If $m \leq C (\log p) (\log \log p) (\log\frac{N}{\log p})$, then $R H D_\xi$ is not a $(p, \epsilon, \eta)-JLE$ for any $0 < \epsilon < 1$.
\end{thm}

\begin{proof}
	For sufficiently large $p$ we can assume that $r := \lfloor \log_2 \log_2 p - 1 \rfloor \in (1, \frac{1}{2}\sqrt{\log_2 p})$. Define $s := 2^r$ and
	\[
	E := \{ D_{\hat{\xi}} \mathbbm{1}_V | V \in \mathbb{G}_{n, r}, \, \hat\xi \in \{+1, -1\}^{N} \}.
	\]
	Every vector in $E$ is $s$-sparse, so there are only $2^s$ different sign patterns corresponding to each $V$. In addition, the number of subspaces is $|\mathbb{G}_{n, r}| \leq 2^{r (n - r)}$ (see \cite{lower_bound_ieee}). Together this gives
	\[
	|E| \leq 2^s |\mathbb{G}_{n, r}| \leq 2^s 2^{r (n - r)} \leq 2^{\frac{1}{2} \log_2 p} 2^{n r} \leq 2^{\frac{1}{2} \log_2(p) + \log_2(N) \frac{1}{2} \sqrt{\log_2 p}} \leq p,
	\]
	using that $\log_2(N) \leq \sqrt{\log_2 p}$ by assumption.
	
	To check the assumptions of Corollary \ref{lower_bound_adapted}, note that $r \geq \log_2 \log_2 p - 2 \geq \log_2(n^c) - 2 \geq 12 \log_2(n)$ for sufficiently large $c$ and $N(\eta)$. Furthermore, $n - r \geq n - \log_2 \log_2 p + 1 \geq n - \log_2 \left( \frac{N}{(\log_2 N)^c} \right) + 1 = n - n + c \log_2 n + 1 \geq 12 \log_2 n$.
	
	By the definition of $s$, we obtain $\frac{1}{4} \log_2(p) \leq s \leq \frac{1}{2} \log_2(p)$ and thus the assumption for $m$ is equivalent to $m \leq C' (\log p)(\log s)(\log\frac{N}{s})$ for a suitable constant $C'$. Together with this upper bound on $m$ and $N \geq N(\eta)$, Corollary \ref{lower_bound_adapted} now implies that with a probability of at least $1 - \eta$ (with respect to the randomness in $R$) there exists a subspace $V \in \mathbb{G}_{n, r}$ such that $R H \mathbbm{1}_V = 0$. For any value of the Rademacher vector $\xi$, we have $x := D_\xi \mathbbm{1}_V \in E$ and $R H D_\xi x = R H \mathbbm{1}_V = 0$. This means that $| \|R H D_\xi x\|_2^2 - \|x\|_2^2 | < \epsilon \|x\|_2^2$ cannot hold for any $\epsilon < 1$. 
\end{proof}

\begin{rem}
	Note that the factor $(\log \log p) (\log \frac{N}{\log p})$ can take values of orders of magnitude between $\log N$ and $(\log N)^2$. Indeed, for the lower bound observe that $\log \log p \geq 2$ and $\log N - \log \log p \geq 2$, so one obtains $(\log \log p) (\log \frac{N}{\log p}) \geq (\log \log p) + (\log \frac{N}{\log p}) = \log N$. The upper bound follows as $\log p = \mathcal{O}(N)$. The maximal order is in fact achieved for $\log p \asymp N^{\alpha}$ where $\alpha$ is an arbitrary constant in $(0, 1)$.
\end{rem}

\subsection{An Impossibility Result for Embeddings into $\ell_1$}

When considering a subsampled Hadamard matrix to be a JLE into $\ell_1$, we obtain a much stronger impossibility result than the one just stated for $\ell_2$. In fact, such matrices cannot, in general, lead to a JLE into $\ell_1$ for any embedding dimension. To see this, again recall that the Hadamard matrix can be interpreted as the Fourier transform on $\mathbb{F}_2^n$. Consider dimensions $N = 2^n$ and $s = 2^r < N$ that are powers of $2$. Take a subspace $V \subset \mathbb{F}_2^n$ of dimension $r$ and $x^{(1)} := \mathbbm{1}_V$. By (\ref{hadamard_subspace}), $H x^{(1)} = \mathbbm{1}_{V^\perp}$. Similarly observe that for $x^{(2)} := \mathbbm{1}_{\{0\}}$, one has that $H x^{(2)} = \mathbbm{1}_{\mathbb{F}_2^n}$.

Now let $R \in \mathbb{R}^{m \times N}$ be a random subsampling matrix in analogy to Corollary \ref{totfjl}. Noting that $H x^{(1)}$ has $\frac{N}{s}$ entries of value $\sqrt\frac{s}{N}$ while $H x^{(2)}$ has $N$ entries of value $\frac{1}{\sqrt{N}}$, we obtain that $\mathbb{E}\|R H x^{(1)}\|_1 = \sqrt\frac{N}{m} m \sqrt\frac{s}{N} \frac{1}{s} = \sqrt\frac{m}{s}$ and $\mathbb{E}\|R H x^{(2)}\|_1 = \sqrt\frac{N}{m} m \frac{1}{\sqrt{N}} = \sqrt{m}$. So although $\|x^{(1)}\|_2 = \|x^{(2)}\|_2 = 1$, there cannot be a scaling factor $\gamma > 0$ such that $\|\gamma R H x^{(1)}\|_1$ and $\|\gamma R H x^{(2)}\|_1$ are concentrated around $1$ at the same time.

Also randomized column signs cannot avoid this problem. To see that, consider the set $E := \{D_{\hat{\xi}} x^{(1)} | \hat{\xi} \in \{\pm 1\}^N \} \cup \{\pm x^{(2)}\}$ which has $2^s + 2$ elements which all have an $\ell_2$ norm of $1$. Then the image $\gamma R H D_\xi E$ will always contain $\gamma R H x^{(1)}$ and $\gamma R H x^{(2)}$, and by the previous paragraph, the norms of the vectors in $\gamma R H D_\xi E$ cannot concentrate around $1$. Thus $\gamma R H D_\xi$ cannot be a JLE for $E$ even for an embedding dimension $m$ as large as $N$.
By Theorem~\ref{thm:l2l1jle}, in contrast, there are $\ell_2 \rightarrow \ell_1$ JLEs of size $m \times N$  for $E$ already if $m \gtrsim \epsilon^{-2} s$ which can be chosen to be much smaller than $N$. That is, subsampled Hadamard matrices with random column signs are intrinsically suboptimal as $\ell_2 \rightarrow \ell_1$ JLEs.

\section{Discussion} \label{sec:discussion}

In this note, we considered a family of Johnson-Lindenstrauss embeddings with order optimal embedding dimension for all possible choices of the parameters $N$, $p$ and $\epsilon$. We showed that it allows for a fast simultaneous embedding of $p'$ points in time $\mathcal{O}(p' N \log m)$, provided that $ p = \exp(\mathcal{O}_\epsilon(N^{1-\delta}))$ for some $\delta>0$ and that $p'$ is large enough, i.e.,~$p' = p$ or $p' \geq N^{d(\delta)}$. This improves upon the previously least restricted fast Johnson-Lindenstrauss transform by Ailon and Liberty \cite{fastjl2}, which required that $p = \exp(\mathcal{O}_\epsilon(N^{\frac{1}{2}-\delta}))$. Our construction incorporates algorithms for fast multiplication of dense matrices.

One may of course ask whether the Hadamard transformation step in our construction is necessary or whether such matrix multiplication algorithms could also be applied directly for a full Bernoulli matrix (as studied, for example, by Achlioptas \cite{binaryjl}). This, however, does not seem to be the case.
Namely, the analogous analysis to our proof above yields for the number of operations the bound of 
\[
\Theta(p' N m^{f(\alpha) - \alpha - 1}).
\]

So if $m = \Theta(N^r)$, i.e.~$p = \exp(\Theta_\epsilon(N^r))$ for any $r > 0$, we obtain a complexity of $\Theta(p' N^{1 + t})$ for a $t > 0$ which is always larger than $\Theta(p' N \log N)$. 
Of course this argument only yields that this particular approach fails, but as the number of entries of $M_{E'}$, one of the matrices to be multiplied, is $N p'$, which is more or less the total complexity bound that we seek, we see very little leverage room for improvements.

Independent from that, it remains a very interesting question to find Johnson-Lindenstrauss embeddings with optimal embedding dimension that allow for a fast query step, i.e., a fast application to a point not part of the given data set. Beyond its use for the approximate nearest neighbor search outlined above, such a result would also yield improvements in other application areas such as the construction of fast transforms with the restricted isometry property as in \cite{jltorip}.

As we found, a subsampled Hadamard matrix with randomized columns signs as considered in \cite{riptojl} cannot achieve this optimal embedding dimension, so a more involved construction will be necessary.

\section*{Acknowledgements}
The authors acknowledge support by the German-Israeli Foundation (GIF) through the grant G-1266-304.6/2015, and by the German science foundation (DFG) in the context of the Emmy Noether junior research group KR 4512/1-1.



\bibliographystyle{unsrt}
\bibliography{Opt_Fast_JL}

\end{document}